\documentclass[a4paper,fleqn]{cas-sc}

\usepackage[authoryear,longnamesfirst]{natbib}
\usepackage{float}
\usepackage{xr}
\usepackage{array,bm,amsmath,amssymb,epsfig,epstopdf}
\usepackage{amsthm}
\usepackage{setspace}
\usepackage{CJK}
\usepackage{color}
\usepackage{lscape}
\usepackage{verbatim}
\usepackage{enumerate}
\usepackage{booktabs}
\usepackage{threeparttable}
\usepackage{multicol}
\usepackage{longtable}
\usepackage{multirow}
\usepackage{bm}
\usepackage{bbm}
\usepackage{cases}
\usepackage{caption}
\usepackage{natbib}
\usepackage{algorithm}
\usepackage{algorithmic}
\usepackage{geometry}
\usepackage{framed}
\usepackage{ragged2e}
\usepackage[figuresright]{rotating}
\usepackage{mathtools}
\usepackage{pifont}
\usepackage{indentfirst}  % 首行缩进的命令
\usepackage{listings}
\definecolor{codegreen}{rgb}{0,0.6,0}
\definecolor{codegray}{rgb}{0.5,0.5,0.5}
\definecolor{codepurple}{rgb}{0.58,0,0.82}
\definecolor{backcolour}{rgb}{0.95,0.95,0.92}
\lstdefinestyle{mystyle}{
backgroundcolor=\color{backcolour},   
commentstyle=\color{codegreen},
keywordstyle=\color{magenta},
numberstyle=\tiny\color{codegray},
stringstyle=\color{codepurple},
basicstyle=\ttfamily\footnotesize,
breakatwhitespace=false,         
breaklines=true,                 
captionpos=b,                    
keepspaces=true,                 
numbers=left,                    
numbersep=5pt,                  
showspaces=false,                
showstringspaces=false,
showtabs=false,                  
tabsize=2
}

\newcommand\RandomVariable[1]{\mathit{1}}

\newcommand\RandomVector[1]{\bm{1}}

\newcommand\Vector[1]{\bm{1}}

\newcommand\MATRIX[1]{\bm{1}}

\newcommand\gN{{\mathcal{N}}}

\newcommand\gS{{\mathcal{S}}}

\newcommand\vecn{\mathrm{vec}}

\DeclareMathOperator*{\argmin}{arg\,min}

\DeclareMathOperator{\Tr}{Tr}

\newtheorem{theorem}{Theorem}[section]
\newtheorem{lemma}{Lemma}[section]

\newtheorem{assumption}{Assumption}[section]
\newcommand{\tabincell}[2]{\begin{tabular}{@{}1@{}}2\end{tabular}}

\def\tsc#1{\csdef{#1}{\textsc{\lowercase{#1}}\xspace}}
\tsc{WGM}
\tsc{QE}
\tsc{EP}
\tsc{PMS}
\tsc{BEC}
\tsc{DE}
\let\printorcid\relax
\begin{document}
\let\WriteBookmarks\relax
\def\floatpagepagefraction{1}
\def\textpagefraction{.001}

\shorttitle{Estimating Precision Matrices for High-Dimensional Interval-Valued Data} 
\shortauthors{Z. Qin et~al.} 

%\begin{frontmatter}

\title [mode = title]{Estimating Precision Matrices for High-Dimensional Interval-Valued Data}                  

\author[1,2]{Zhongfeng Qin}
\ead{qin@buaa.edu.cn}
\author[1]{Hao Xu}
\ead{haoxu124@buaa.edu.cn}
\author[1]{Wenhao Cui}
\ead{wenhaocui1992@buaa.edu.cn}
\author[3,4]{Wan Tian}
\ead{wantian61@foxmail.com}
\cormark[1] 
\affiliation[1]{organization={School of Economics and Management, Beihang University},
                city={Beijing},
                postcode={100191},
                country={China}}
\affiliation[2]{organization={Key Laboratory of Complex System Analysis, Management and Decision (Beihang University)},
                city={Beijing},
                postcode={100191},
                country={China}}
\affiliation[3]{Advanced Institute of Information Technology, Peking University}    

\affiliation[4]{organization={Wangxuan Institute of Computer Technology, Peking University},
                city={Beijing},
                postcode={100871},
                country={China}}       
\cortext[cor1]{Corresponding author} 

\begin{abstract}
In the field of statistical learning and data analysis, estimating precision matrices (i.e., the inverse of covariance matrices) is a critical task, particularly for understanding dependency structures among variables. However, traditional methods often fall short when dealing with high-dimensional interval-valued data, where each observation is represented as an interval rather than a single point. This paper proposes a novel framework for estimating precision matrices in such contexts, addressing the unique challenges posed by the interval nature of the data. Specifically, we assume that the upper and lower bounds of the intervals share the same conditional dependency structure, and then formulate the interval graphical lasso optimization objective to estimate the precision matrix. At the optimization level, we provide an efficient computational approach, while at the theoretical level, we prove the sparsity and consistency of the estimator. Experimental results on simulated studies and real data applications demonstrate the superiority of the proposed method in terms of estimation precision and interpretability.
\end{abstract}

% \begin{highlights}
% \item 
% \item 
% \item 
% \end{highlights}

\begin{keywords}
Precision matrix \sep Interval-valued data\sep Regularization\sep  High-dimensional \sep Graphical lasso
\end{keywords}
\maketitle

\section{Introduction} \label{sec1}
The modeling and inference of interval-valued data, particularly in the context of interval-valued time series, have garnered significant attention in the fields of statistics and econometrics \citep{billard2003statistics, gonzalez2013constrained,  golan2015interval, han2016vector, sun2022model}. Interval-valued data are prevalent in practical applications. For instance, in meteorological science, daily temperature ranges, defined by the maximum and minimum temperatures, constitute interval observations. Other notable examples include the annual GDP growth rate range of a country, represented by its maximum and minimum values, and the blood pressure interval, which encompasses both diastolic and systolic measurements. Compared to point-valued data, interval-valued data offer two key advantages \citep{han2016vector}. First, interval observations inherently capture both level and variability, providing richer information. This enhanced informativeness enables more efficient estimation and more powerful statistical inference. Second, from a robustness standpoint, specific perturbations can significantly distort the statistical inference of point-valued data, whereas interval-valued data modeling mitigates this issue, offering greater resilience to such disturbances \citep{chou2005forecasting}.

% 当前，关于区间值数据的研究已涵盖广泛领域，包括但不限于回归分析、区间值时间序列分析、可视化及模型平均。然而，若要深入挖掘多元区间变量间潜在的条件依赖结构，核心在于对精度矩阵的有效估计。在统计学习与数据分析任务中，精度矩阵不仅是构建高斯图模型以揭示变量间条件独立性的数学基础，更是构建线性判别分析（LDA）分类器的关键参数。此外，该矩阵在基因调控网络推断及金融领域的最小方差投资组合构建中也发挥着不可替代的作用。尽管精度矩阵估计在经典点值数据分析中已被证实其重要价值，但针对区间值数据的精度矩阵研究仍极为有限，且缺乏统一的理论框架。下文将分别探讨该领域在低维与高维场景中的挑战与进展。
Current research on interval-valued data encompasses a wide range of topics, including but not limited to regression analysis \citep{billard2000regression, gonzalez2013constrained, lima2017nonlinear, sun2016linear}, interval-valued time series analysis \citep{maia2008forecasting, maia2011holt, sun2018threshold}, visualization \citep{le2012symbolic, zhang2022visualization}, and model averaging \citep{sun2022model}. However, to further explore the latent conditional dependence structures among multivariate interval-valued variables, the central task lies in the effective estimation of the precision matrix. In statistical learning and data analysis, the precision matrix not only serves as the mathematical foundation for constructing Gaussian graphical models to reveal conditional independence relationships among variables, but also plays a central role in the construction of linear discriminant analysis (LDA) decision rules. Moreover, it plays an indispensable role in gene regulatory network inference and in the construction of minimum-variance portfolios in financial applications\citep{schafer2005empirical,friedman2008sparse,kan2024sample}. Although the importance of precision matrix estimation has been well established in the analysis of point-valued data, research on precision matrices for interval-valued data remains extremely limited and lacks a unified theoretical framework.
% As we know, the effective estimation of covariance and precision matrices is a critical step in numerous statistical learning and data analysis tasks.However, research on covariance and precision matrices for interval-valued data remains extremely limited and lacks a unified framework.
Below, we discuss the challenges and advancements in both low-dimensional and high-dimensional settings.

In low-dimensional settings, \citet{cazes1997extension} proposed principal component analysis methods based on vertices and centers for interval-valued data, simultaneously introducing two covariance matrix estimators. Nevertheless, both estimators utilize only a portion of the information contained within the interval observations. \citet{wang2012cipca} assumed an infinite number of uniformly distributed points within the hypercube and defined the inner product operator for interval-valued variables, as well as the corresponding covariance. \citet{bertrand2000descriptive} defined the empirical distribution function and density function for interval-valued data, laying the groundwork for the empirical mean and variance. Building on this work, \citet{billard2008sample} further extended the definition to the covariance of bivariate interval-valued variables. To enhance the robustness of the estimator, \citet{tian2024minimum} extended the minimum covariance determinant (MCD) estimator \citep{rousseeuw1985multivariate} to the interval-valued data scenario, using the covariance matrix defined by \citet{billard2008sample} as a baseline. Additionally, they proposed a computationally efficient algorithm for implementation. On the theoretical front, they also established the high breakdown point property of the proposed estimator, demonstrating its robustness against outliers and data contamination.
In high-dimensional settings, to the best of our knowledge, only two studies have addressed matrix estimation for interval-valued data. The first focuses on covariance matrix estimation, where \citet{tian2024minimum}, within the framework of the MCD, proposed regularized and projection-based methods for estimating high-dimensional covariance matrices. The second study pertains to precision matrix estimation, where \citet{wu2024idgm} employed a bivariate point-value representation of intervals and utilized graphical lasso \citep{friedman2008sparse} to compute the precision matrix of the bivariate points as a measure of conditional dependency for interval-valued variables. However, this approach suffers from a critical limitation: it doubles the dimensionality of the matrix, resulting in reduced interpretability and counterintuitive outcomes. Therefore, it is both urgent and essential to develop a precision matrix estimation method that aligns with the dimensionality of interval-valued variables.

In contrast to interval-valued data, the theory and methodologies for matrix estimation of point-valued data have been progressively refined and well-established. In low-dimensional settings, most research has focused on robust estimation of covariance matrices. One of the most widely used multivariate high-breakdown-point methods is the Minimum Volume Ellipsoid (MVE) estimator \citep{van2009minimum}, which aims to find the ellipsoid of minimal volume that encompasses at least half of the observations. \citet{rousseeuw1985multivariate} introduced a highly robust covariance matrix estimation method known as the MCD estimator, whose core principle is to select a subset of \(h\) observations out of \(n\) to estimate the covariance matrix with the smallest determinant. Theoretically, MCD exhibits superior properties compared to MVE; however, due to its computational complexity, it was not widely adopted until \citet{rousseeuw1999fast} proposed the Fast-MCD algorithm. In high-dimensional settings, most studies rely on regularization techniques and structural assumptions about the target matrix to achieve consistent estimation. These widely adopted assumptions encompass sparsity \citep{fan2016overview}, banding structures \citep{wu2003nonparametric, Bickel2008RegularizedEO}, and tapering \citep{Furrer2007EstimationOH, Cai2010OptimalRO}. By combining a predefined symmetric positive definite target matrix with the covariance matrix of a standardized subset of \(h\) observations, \citet{boudt2020minimum} extended MCD to high-dimensional scenarios. In the realm of high-dimensional precision matrix estimation, graphical lasso plays a pivotal role due to its ability to efficiently recover sparse dependency structures, making it a cornerstone method for large-scale statistical inference. For a comprehensive discussion on matrix estimation methods for high-dimensional point-valued data, readers are referred to \cite{fan2016overview} and the references included therein.

% Building upon the precision matrix estimation methods for high-dimensional point-valued data, we propose a novel framework specifically tailored for interval-valued data.	Specifically, we assume that the upper and lower bounds of interval-valued variables share the same dependency structure, i.e., the same precision matrix and covariance matrix. Consequently, we derive the joint likelihood of the upper and lower bounds and incorporate a sparsity penalty to formulate an optimization problem for estimating the precision matrix. This results in an optimization objective similar to graphical lasso, which we term interval graphical lasso (IGL). On the optimization front, we provide an efficient computational algorithm; on the theoretical front, under mild conditions, we establish the sparsity and consistency of the estimator. Through examining different dependency structures in simulation studies, as well as real-world applications in portfolio construction, we demonstrate the effectiveness and interpretability of the proposed estimator.

Building upon the precision matrix estimation methods for high-dimensional point-valued data, we propose a novel framework specifically tailored for interval-valued data.	In many practical applications, such as financial markets, the upper and lower bounds of interval data essentially characterize the fluctuation boundaries of a variable over a specific observation period. Taking stock prices as an example, the interdependence among different assets primarily stems from their response to common market shocks and macroeconomic information. The daily high and low prices of stocks often exhibit consistent conditional dependencies. Based on this, we assume that the upper and lower bounds of interval-valued variables share the same dependency structure, i.e., the same precision matrix and covariance matrix. Consequently, we derive the joint likelihood of the upper and lower bounds and incorporate a sparsity penalty to formulate an optimization problem for estimating the precision matrix. This results in an optimization objective similar to graphical lasso, which we term interval graphical lasso (IGL). We analyze IGL from both algorithmic and theoretical perspectives. In terms of computational efficiency, we propose an efficient block coordinate descent framework for solving the objective function. Theoretically, under mild assumptions, we prove the sparsity and consistency of the estimator. This implies that IGL can effectively identify zero elements in the precision matrix and recover the true dependency structure of the interval variables.
Through examining different dependency structures in simulation studies, as well as developing investment strategies with higher Sharpe ratios in an empirical analysis of S\&P 500 component portfolios, we have comprehensively validated the superiority and practical value of the IGL method.

The structure of this paper is organized as follows: Section \ref{sec2} introduces the proposed IGL and provides the corresponding optimization algorithm. Section \ref{sec3} discusses the theoretical properties of the estimator. Section \ref{sec4} presents simulation studies,  which primarily analyzes the precision matrices of different sparse structures and their estimation accuracy. Section \ref{sec5} applies the proposed estimator to real-world applications in portfolio construction to illustrate its interpretability advantages. Section \ref{sec6} concludes the paper. All theoretical proofs are provided in the Appendix \ref{appendixA}.

\section{Methodology}\label{sec2}
In this section, we present the proposed IGL and provide an efficient algorithm for its solution. 

Let \(X = (X_{(1)}, X_{(2)}, \cdots, X_{(p)}) = (X_1,X_2,\cdots, X_n)^\top\) be the interval-valued observation matrix, where \(X_{(j)} = (X_{1j}, X_{2j}, \cdots, X_{nj})^\top\) represents the \(n\) observations of the \(j\)-th interval-valued variable, and \(X_i = (X_{i1}, X_{i2}, \cdots, X_{ip})^\top\) denotes the \(i\)-th observation of the \(p\) interval-valued variables. Each interval-valued observation unit \(X_{ij} = [X^{l}_{ij}, X^{u}_{ij}]\), where \(X^{l}_{ij}\) and \(X^{u}_{ij}\) are the upper and lower bounds of the interval, respectively. Our current objective is to estimate the precision matrix corresponding to the \( p \) interval-valued variables based on the interval-valued observation matrix \(X\).

In our proposed method, we begin by assuming that the upper and lower bounds of the intervals share the same dependency structure and are independently and identically distributed (i.i.d.) according to a normal distribution. This assumption can be formally expressed as follows:
\[
(X^l_i)^n_{i=1} \overset{i.i.d.}{\sim} \gN(\mu^l, \Sigma), \quad \ 	 (X^u_i)^n_{i=1} \overset{i.i.d.}{\sim} \gN(\mu^u, \Sigma), 
\]
where \(X^l_i = (X^l_{i1}, X^l_{i2}, \cdots, X^l_{ip})^\top\) and \(X^u_i = (X^u_{i1}, X^u_{i2}, \cdots, X^u_{ip})^\top\) represent the upper and lower bounds of the \( i \)-th observation \(X_i\), respectively. This assumption is motivated by two key reasons. First, for interval-valued data, the upper and lower bounds typically exhibit similar trends. For instance, in the context of stock prices, when the daily high price increases, the daily low price also tends to rise. Figure 1 provides an illustrative example of this phenomenon. Second, this assumption allows us to seamlessly extend the theories and methodologies developed for point-valued data to the interval-valued scenario.

\begin{figure}[pos=htbp]
\centering 
\begin{minipage}[t]{0.44\textwidth}
\centering
\includegraphics[scale=0.25]{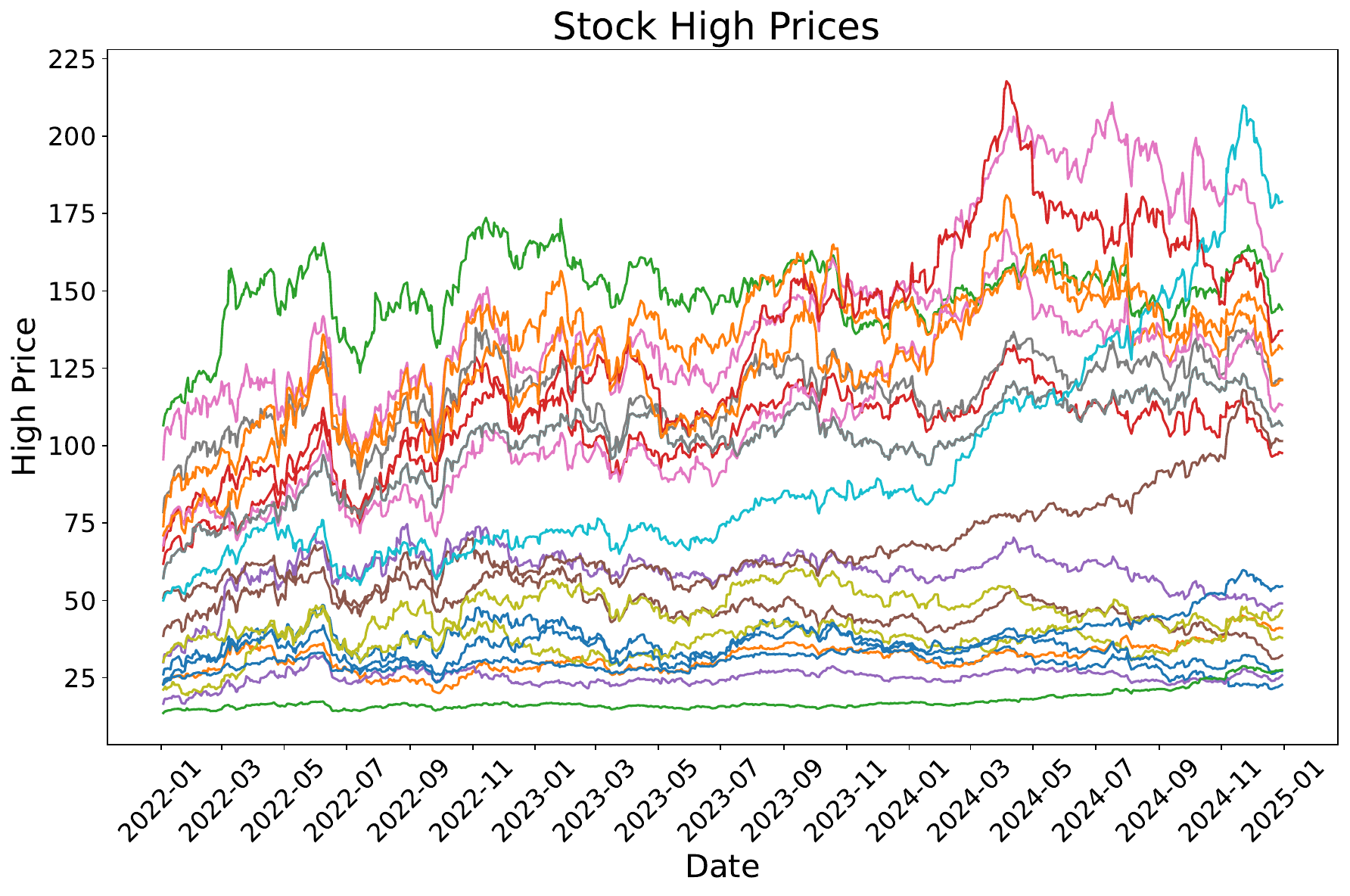}
\end{minipage}
\begin{minipage}[t]{0.27\textwidth}
\centering
\includegraphics[scale=0.23]{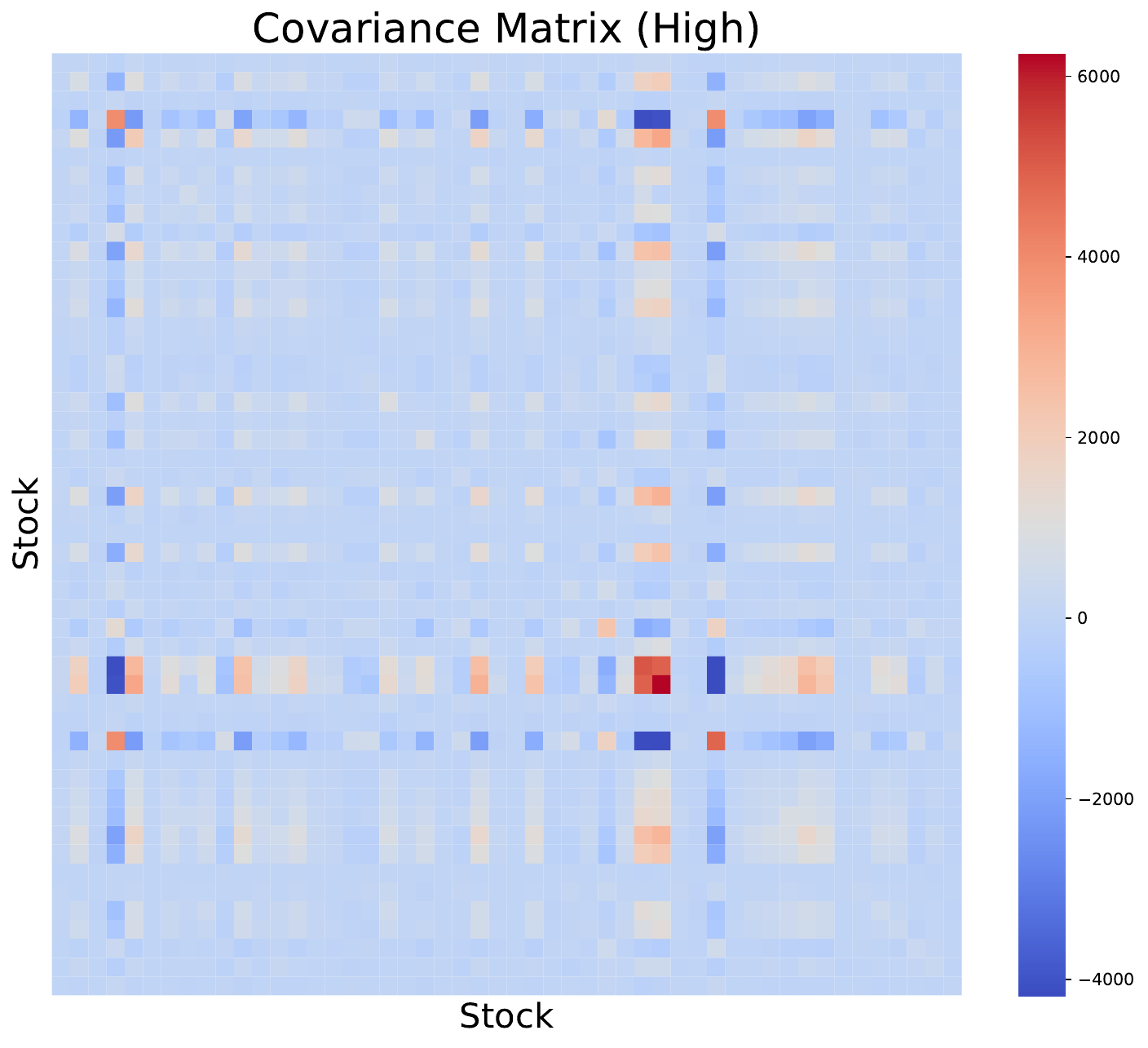}
\end{minipage}
\begin{minipage}[t]{0.27\textwidth}
\centering
\includegraphics[scale=0.23]{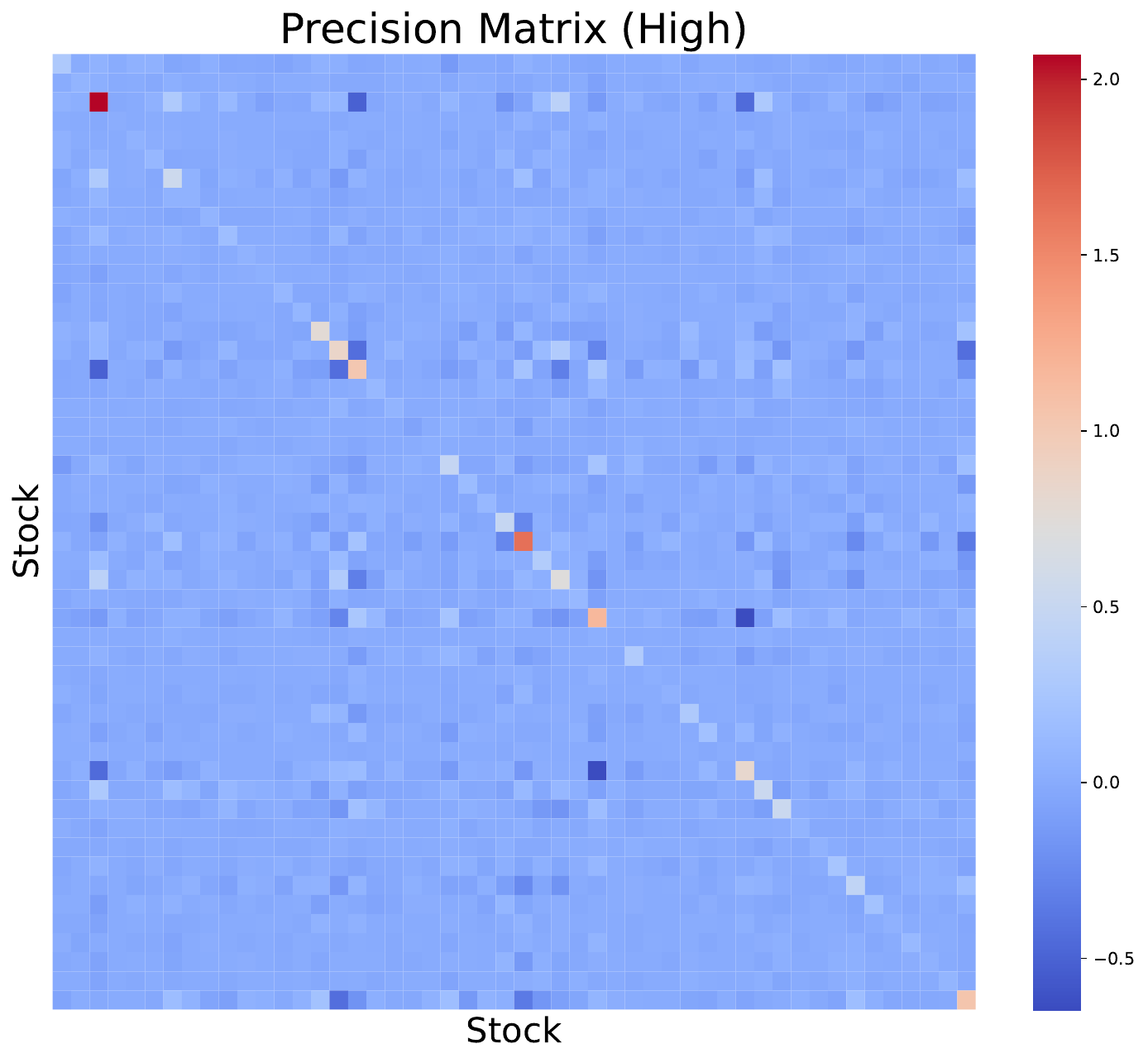} 
\end{minipage}
\begin{minipage}[t]{0.44\textwidth}
\centering
\includegraphics[scale=0.25]{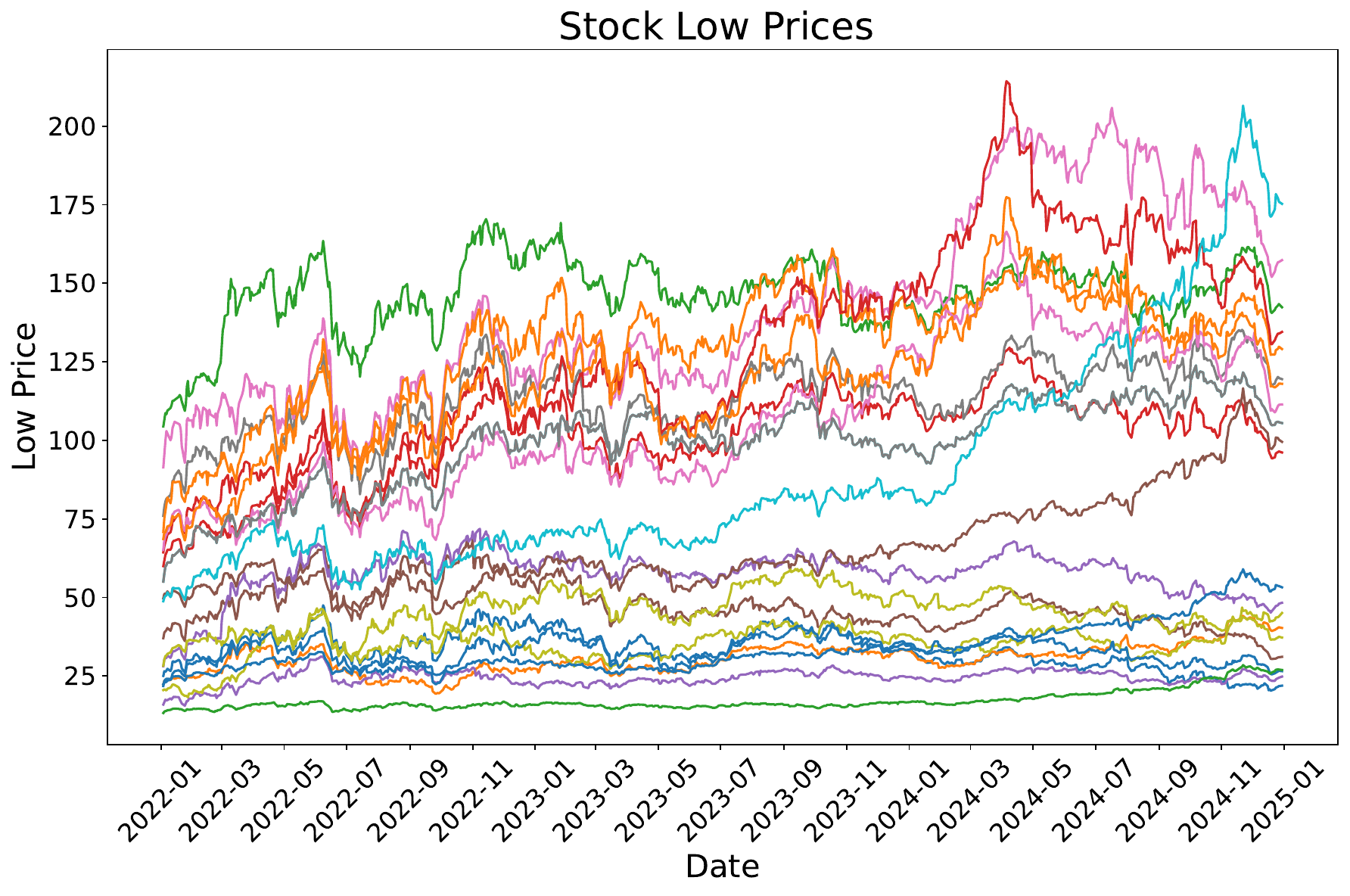}
\end{minipage}
\begin{minipage}[t]{0.27\textwidth}
\centering
\includegraphics[scale=0.23]{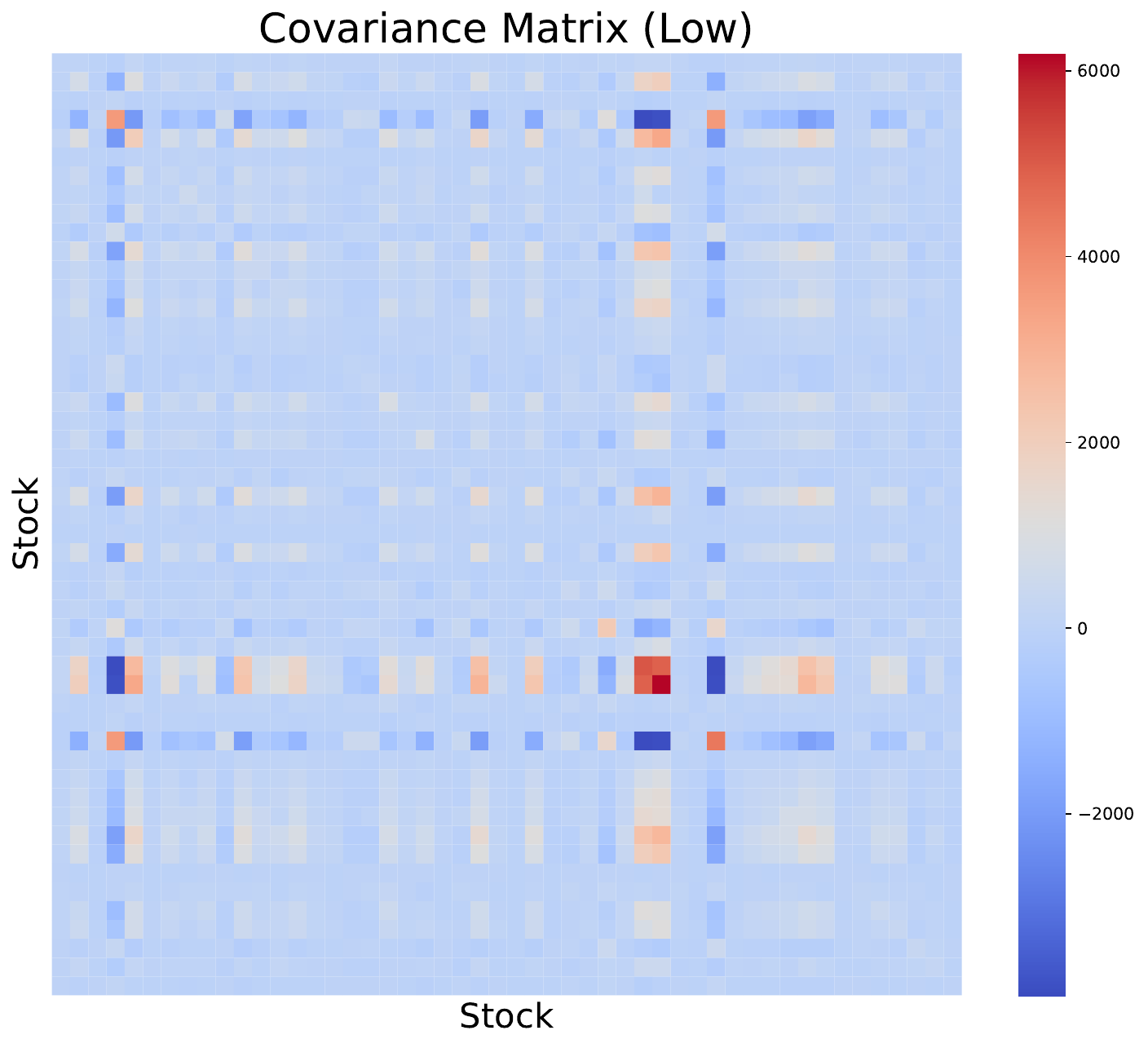}
\end{minipage}
\begin{minipage}[t]{0.27\textwidth}
\centering
\includegraphics[scale=0.23]{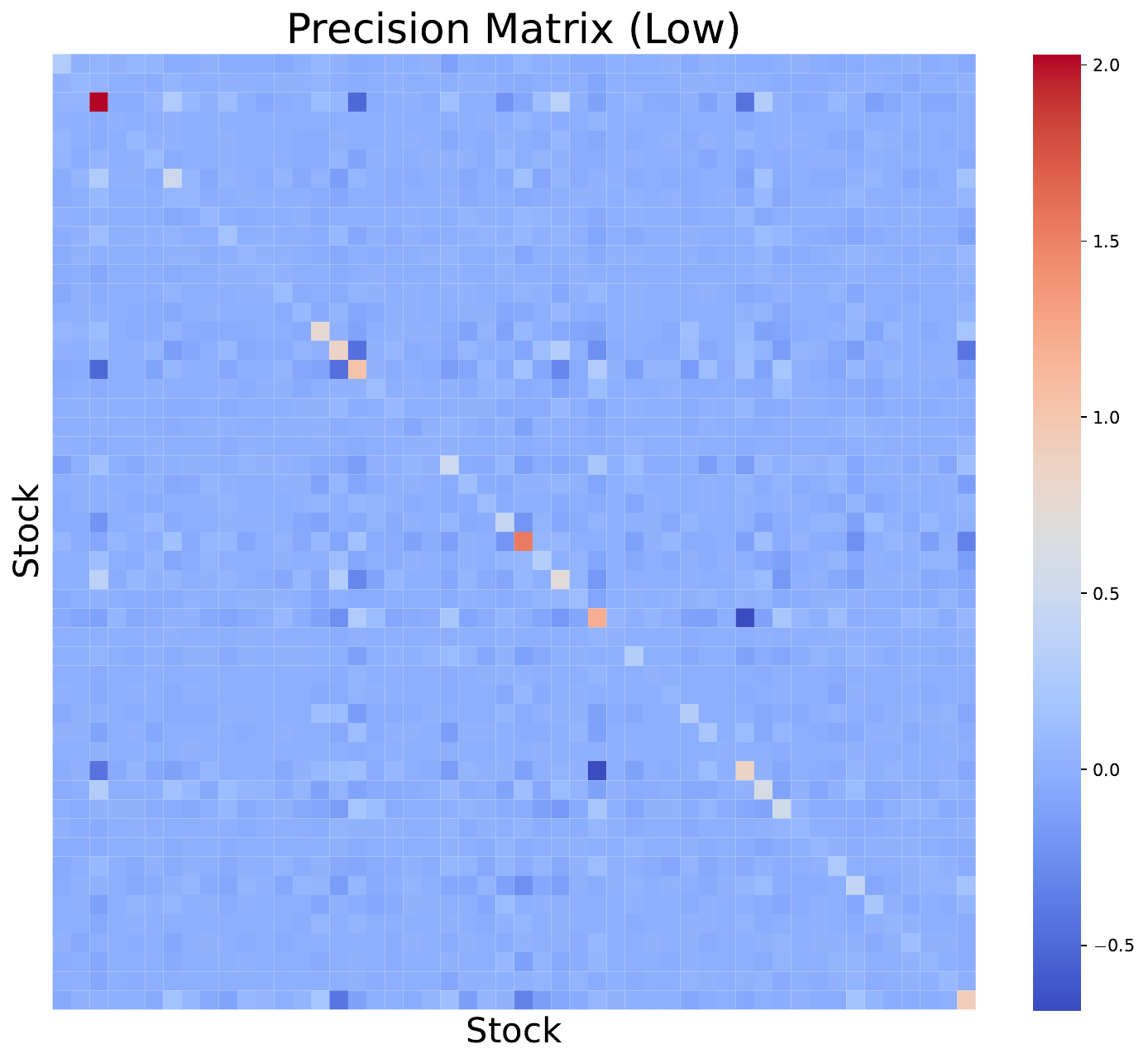} 
\end{minipage}
\caption{We visualized the data of 50 stocks from the Energy sector of the S\&P 500 Index during the period from January 1, 2022, to December 31, 2024, including trend charts, covariance matrix heatmaps, and precision matrix heatmaps. The two rows correspond to the daily highest and lowest prices of the stocks, respectively. The second column demonstrates that estimating the covariance matrix based on the highest and lowest prices yields similar results, while the third column shows that the precision matrix is also analogous. This supports our hypothesis.}\label{spthreesy}
\end{figure}

Under this assumption, we can directly express the negative log-likelihood function for the upper and lower bounds of each observation as follows:
\begin{equation}\label{loglike}
\begin{aligned}
\ell^l(X^l_i, \Theta) &= \frac{1}{2}((X^l_i - \mu^l_i)^\top \Theta (X^l_i - \mu^l_i) - \log \det(\Theta) + p\log(2\pi)), i = 1,2,\cdots, n,\\
\ell^u(X^u_i, \Theta) &= \frac{1}{2}((X^u_i - \mu^u_i)^\top \Theta (X^u_i - \mu^u_i) - \log \det(\Theta) + p\log(2\pi)), i = 1,2,\cdots, n,\\
\end{aligned}
\end{equation}
where \(\Theta = (\theta_{ij})_{1\leq i,j\leq p}= \Sigma^{-1}\) is the precision matrix, and \( \det(\cdot)\) denotes the determinant operation. When the data deviate from a normal distribution or exhibit dependencies, we refer to (\ref{loglike}) as a quasi log-likelihood function. Similar to the classical methods for estimating precision matrices of high-dimensional point-valued data, we also employ a penalized likelihood approach for estimation. Specifically, the proposed IGL takes the following form:
\begin{equation} \label{penlike}
p_\lambda(\Theta) + \sum_{i=1}^{n}  \ell^l(X^l_i, \Theta) + \ell^u(X^u_i, \Theta),
\end{equation}
where \(p_\lambda(\Theta)\) represents the penalty function, and \(\lambda \) is the regularization parameter used to control the sparsity level of \( \Theta \). In this paper, we set \(p_\lambda(\Theta) = \lambda \lVert \Theta \rVert_{1} = \sum_{i, j} \lambda |\theta_{ij}|\). Naturally, there are many other penalty functions available as alternatives, such as smoothly clipped absolute deviation (SCAD) penalty \citep{fan2001variable} , minimax concave penalty (MCP) \citep{zhang2010nearly}, and others.

We then discuss the computational aspects of the optimization objective (\ref{penlike}) corresponding to the IGL. Through some basic algebraic manipulations and omitting constants that are independent of \(\Theta\), the optimization objective (\ref{penlike}) is equivalent to:
\begin{equation} \label{finalobj}
\begin{aligned}
p_\lambda(\Theta) &+ \sum_{i=1}^{n}\frac{1}{2}(X^l_i - \mu^l_i)^\top \Theta (X^l_i - \mu^l_i) + \frac{1}{2}(X^u_i - \mu^u_i)^\top \Theta (X^u_i - \mu^u_i)- n\log \det(\Theta)\\
& = p_\lambda(\Theta)  + \frac{1}{2}\sum_{i=1}^{n}\left((X^l_i - \mu^l_i)^\top \Theta (X^l_i - \mu^l_i) + (X^u_i - \mu^u_i)^\top \Theta (X^u_i - \mu^u_i)\right) - n\log \det(\Theta)\\
& = p_\lambda(\Theta)  + \frac{n}{2} \Tr(S^l\Theta ) + \frac{n}{2} \Tr( S^u\Theta) - n\log \det(\Theta)\\
& \propto \ell(\Theta) \coloneqq \frac{2}{n}p_\lambda(\Theta)+ \Tr((S^l + S^u)\Theta ) -2\log \det(\Theta),
\end{aligned}
\end{equation}
where \(S^l\) and \(S^u\) are the empirical covariance matrices of the upper and lower bounds of the intervals, respectively, and \(\Tr(\cdot)\) denotes the matrix trace operator. Note that the term \(2/n\) can be incorporated into the regularization by appropriately scaling \( \lambda \). Therefore, for notational simplicity, we typically omit this term in (\ref{finalobj}) and adjust \( \lambda \) accordingly. Our final optimization problem is
\begin{equation}\label{finalobjnew}
\min_{\Theta \in \gS^p_{++}} \ell(\Theta) = p_\lambda(\Theta)+ \Tr((S^l + S^u)\Theta ) -2\log \det(\Theta),
\end{equation}
where \(\gS^p_{++} \) denotes the set of \(n\times n\) symmetric positive definite matrices. We employ the block coordinate descent algorithm proposed by \citet{banerjee2008model} to solve optimization problem (\ref{finalobjnew}). First, we directly observe that \(\lVert \Theta \rVert_{1} = \max_{\lVert P \rVert_{\infty} \leq 1} \Tr(\Theta P)\), where \(P\) is a symmetric matrix, and \(\lVert P \rVert_{\infty}\) denotes the maximum absolute value element of \(P\). The original optimization problem (\ref{finalobjnew}) can be formulated as: 
\begin{equation}\label{finalobjnew2}
\max_{\Theta \in \gS^p_{++}} \min_{\lVert P \rVert_{\infty} \leq \lambda} 2\log \det(\Theta) - \Tr(\Theta P) - \Tr((S^l + S^u)\Theta ).
\end{equation}

We can obtain the dual problem corresponding to (\ref{finalobjnew2}) by swapping the min and max. The inner optimization problem with respect to \(\Theta\) can be directly solved using the first-order optimality condition, i.e., \(\Theta = ((P+S^l + S^u)/2)^{-1}\). By disregarding some constants, we can conclude that (\ref{finalobjnew2}) is equivalent to \(\min_{\lVert P \rVert_{\infty} \leq \lambda} -2 \log \det((P+S^l + S^u) / 2) - 2p\). Let the working matrix \(W = (P+S^l + S^u) / 2\), we obtain the dual problem as follows:
\begin{equation}\label{finalobjnew3}
\widehat{\Sigma} \coloneqq \max \{\log \det(W): \lVert 2W-S^l - S^u \rVert_{\infty} \leq \lambda\}.
\end{equation}

From this, it is clear that the primal optimization problem involves estimating the precision matrix, while its dual problem pertains to estimating the covariance matrix. Prior to presenting the block coordinate descent algorithm for solving the dual problem (\ref{finalobjnew3}), we first introduce some fundamental notations. For any symmetric matrix \(A = (a_{ij})_{1 \leq i,j \leq p} \in \mathbb{R}^{p\times p}\), let \(A_{-i,-j} \in \mathbb{R}^{(p-1)\times (p-1)}\) denote the matrix obtained by removing the \(i\)-th row and \(j\)-th column of \(A\). Similarly, \(A_{j} \in \mathbb{R}^{p-1}\) represents the vector formed by the \(j\)-th column of matrix \(A\) with its diagonal element \(a_{jj}\) removed. For example, a matrix \(A\) can be expressed as 
\[
A = \left(
\begin{array}{cc}
A_{-p,-p} & A_{p}\\
A^\top_{p} & a_{pp}
\end{array}
\right).
\]

Below is the block coordinate descent algorithm for solving the dual problem (\ref{finalobjnew3}).

\begin{algorithm}[H]
\caption{Block coordinate descent.}
\label{BCDalgo}
\begin{algorithmic}
\REQUIRE Regularization parameter \(\lambda\), precision tolerance \(\epsilon\).\\
\ENSURE \(W^{(0)} = (\lambda I +S^l + S^u) / 2\).\\
Step 1: For \(j = 1,2,\cdots, p\):\\
\quad \quad Step 1.1: Solve the quadratic program
\begin{equation}\label{partquadratic}
\widehat{V} \coloneqq \argmin_{V \in \mathbb{R}^{p-1}} \{V^\top (W^{(j-1)}_{-j,-j})^{-1}V:\lVert 2V-S^l_j - S^u_j \rVert_{\infty} \leq \lambda \}.
\end{equation}
\quad \quad  Step 1.2: Update the row and column \(W^{(j-1)}_j\) in \(W^{(j-1)}\) to \(\widehat{V}\), resulting in \(W^{(j)}\).\\

Step 2: Check whether the following convergence condition is satisfied.
\begin{equation} \label{convergencecon}
\Tr\left((S^l + S^u)(W^{(p)})^{-1}\right) -2p + \lambda \lVert (W^{(p)})^{-1} \rVert_{1} \leq \epsilon.
\end{equation}

Set \(W^{(0)} \longleftarrow W^{(p)}\), and repeat the above steps until convergence.\\

\textbf{Output}: The estimated covariance matrix \(\widehat{\Sigma} \).
\end{algorithmic}
\end{algorithm}

After obtaining the estimated covariance matrix using Algorithm \ref{BCDalgo}, its inverse is computed to yield the final precision matrix. Furthermore, the optimization problem (\ref{partquadratic}) can be essentially reformulated as an equivalent penalized least squares problem \citep{banerjee2008model, friedman2008sparse}, which can be solved efficiently. From the fact that \(\lVert \Theta \rVert_{1} = \max_{\lVert P \rVert_{\infty} \leq 1} \Tr(\Theta P)\), it can be observed that the essence of the convergence criterion (\ref{convergencecon}) is to minimize the first two terms of the optimization objective (\ref{finalobjnew}) as much as possible. We proceed to discuss the uniqueness of the solution to the optimization problem (\ref{finalobjnew}) as well as the convergence of the solution obtained using Algorithm \ref{BCDalgo}.

\begin{theorem}\label{uniquesolu}
Let \(\lambda_{\max} (A)\) denote the largest eigenvalue of matrix \(A\). For any \(\lambda >0\), the optimal solution to the optimization problem (\ref{finalobjnew}) is unique, and its largest eigenvalue satisfies
\[
\left(\frac{\lambda p}{2} + \lambda_{\max}((S^l + S^u)/2)\right)^{-1} \leq \lambda_{\max}(\widehat{\Sigma}^{-1}) \leq \frac{2p}{\lambda}.
\]
\end{theorem}

In high-dimensional scenarios, to establish the consistency of the estimated precision matrix, we typically assume that the smallest and largest eigenvalues of the true precision matrix are bounded \citep{lam2009sparsistency}. Theorem \ref{uniquesolu} demonstrates that our estimated precision matrix aligns with this assumption, while also indicating that our solution based on the block coordinate descent algorithm exhibits numerical stability. 

\begin{theorem}\label{optimicon}
By employing block coordinate descent to solve the optimization problem (\ref{finalobjnew3}), it is ensured to converge and yield an \(\epsilon\)-optimal solution. Moreover, the estimator obtained at each iteration is strictly positive definite, i.e., \(W^{(j)} \succ 0, j=1,2,\cdots, p\).
\end{theorem}

\section{Theoretical Properties}\label{sec3}
In this section, we investigate the asymptotic properties of the estimator obtained from the proposed IGL method. We establish the consistency in the Frobenius norm and the selection consistency (sparsistency) of the estimated precision matrix under the high-dimensional scaling regime where the dimension \(p\) may grow with the sample size \(n\).

Let \(\Theta_0\) denote the true precision matrix, and \(\Sigma_0=\Theta_0^{-1}\) be the true covariance matrix. We define the active set of the true precision matrix as \(S=\{(i,j):(\Theta_0)_{ij}\neq0\}\), and its complement as \(S^c=\{(i,j):(\Theta_0)_{ij}=0\}\). Let \(s=|S|\) denote the total cardinality of the support, and \(s_{\text{off}}=|\{(i,j)\in{S}:i\neq{j}\}|\) denote the number of non-zero off-diagonal elements.
For any matrix \(A\), let \(\lambda_{\min}(A)\) and \(\lambda_{\max}(A)\) denote the minimum and maximum eigenvalues, respectively. We define the element-wise infinity norm as \(\|A\|_{\max}=\max_{i,j}|A_{ij}|\), the Frobenius norm as \(\|A\|_F=(\sum_{i,j}A_{ij}^2)^{1/2}\), and the spectral norm as \(\|A\|_2\). Recall that our method utilizes the pooled empirical covariance matrix, which we define as \(\overline{S}=(S^l+S^u)/2\). Accordingly, we define the rescaled regularization parameter as \(\widetilde{\lambda}=\lambda/2\).

To facilitate the theoretical analysis, we impose the following regularity conditions, which are standard in the literature on high-dimensional graphical model estimation \citep{ravikumar2011high, rothman2008sparse}.

\begin{assumption}[Eigenvalue Boundedness]\label{ass:wellcond}
The true precision matrix \(\Theta_0\) is well-conditioned. Specifically, there exist constants \(0<\tau_{\min}\le\tau_{\max}<\infty\) such that
\[
\tau_{\min}\le\lambda_{\min}(\Theta_0)\le\lambda_{\max}(\Theta_0)\le\tau_{\max}.
\]
\end{assumption}

\begin{assumption}[Concentration of Pooled Covariance]\label{ass:conc}
The pooled sample covariance matrix \(\overline{S}\) concentrates around the true covariance \(\Sigma_0\). Specifically, there exist constants \(c_1,c_2>0\) such that for any \(\delta>0\),
\[
\mathbb{P}(\|\overline{S}-\Sigma_0\|_{\max}\ge\delta)\le{c_1p^2\exp(-c_2n\delta^2)}.
\]
\end{assumption}
\noindent \textit{Remark:} Since the upper and lower bound observations \(X^l\) and \(X^u\) are assumed to be sub-Gaussian, the pooled estimator \(\overline{S}\) naturally satisfies this concentration inequality with \(\delta\asymp\sqrt{\log{p}/n}\).

\begin{assumption}[Regularization Regime]\label{ass:lambda}
The regularization parameter \(\widetilde{\lambda}\) is chosen to dominate the noise level. We assume
\[
\widetilde{\lambda}\asymp\sqrt{\frac{\log{p}}{n}},
\]
which ensures that \(\widetilde{\lambda}\ge{C}\|\overline{S}-\Sigma_0\|_{\max}\) holds with high probability for some constant \(C>2\).
\end{assumption}

Based on the assumptions above, we derive the convergence rates for the IGL estimator. The following theorem establishes the estimation consistency in terms of the Frobenius norm.

\begin{theorem}[Frobenius Norm Consistency]\label{thm:rate}
Suppose Assumptions \ref{ass:wellcond}--\ref{ass:lambda} hold. Let \(\widehat{\Theta}\) be the global minimizer of the objective function (\ref{finalobjnew}). Then, with probability tending to one, the estimation error satisfies:
\[
\|\widehat{\Theta}-\Theta_0\|_F=O_P\left(\sqrt{\frac{(p+s)\log{p}}{n}}\right).
\]
\end{theorem}

Theorem \ref{thm:rate} indicates that consistent estimation is achievable even when \(p\gg{n}\), provided that the true precision matrix is sufficiently sparse (i.e., \(s\) is small relative to \(n\)). We next address the model selection consistency, specifically the property that the estimator does not falsely select zero entries as non-zero (sparsistency).

\begin{theorem}[Sparsistency]\label{thm:sparse}
Suppose Assumptions \ref{ass:wellcond}--\ref{ass:lambda} hold. Let \(\widehat{\Theta}\) be a local minimizer satisfying the rate of convergence given in Theorem \ref{thm:rate}. If the regularization parameter satisfies \(\widetilde{\lambda}\gg\sqrt{\log{p}/n}\), then with probability approaching one,
\[
\widehat{\Theta}_{ij}=0\quad\text{for all}\quad(i,j)\in{S^c}.
\]
\end{theorem}

\noindent Theorem \ref{thm:sparse} guarantees that the IGL estimator can correctly recover the zero pattern of the precision matrix, thereby consistently reconstructing the graph structure of the interval-valued variables.

\section{Simulation Studies}\label{sec4}
% section 4 Simulation Studies不必要分多个subsection，直接一口气写完；4 Simulation Studies的符号与正文统一，不必要很多内容加粗；BIC准则应该先说标准形式，再说明是如何推广到区间的；performance指标需要写出具体的计算公式；Figure 2 3 4按照figure 1，调整字体字号，多图没有对齐；Table 1 2 3的第一列列名有问题；结果分析：先说明结果在哪些图表呈现，再依次分析，再写部分总结性和结论性的话；读读论文的证明；
In this section, we evaluate the numerical performance of the proposed IGL method in finite-sample settings. We detail the data-generating framework, including the construction of interval-valued data and precision matrices with varying sparsity patterns, followed by a comprehensive analysis of the estimation accuracy under different high-dimensional scenarios.

Data generation proceeds in two stages: generating latent multivariate normal data and transforming them into intervals. Let $Z=(z_{ij})\in\mathbb{R}^{n\times{p}}$ be the latent data generated from $\mathcal{N}_p(\bm{0},\bm{\Sigma})$. We employ three mechanisms to construct the interval-valued data: DGP1 (Fixed Shift) sets $[z_{ij},z_{ij}+C]$ with a constant width $C>0$. DGP2 (Symmetric Fixed Width) constructs $[z_{ij}-C/2,z_{ij}+C/2]$. DGP3 (Random Width) introduces heterogeneity by setting $[z_{ij}-r_i,z_{ij}+r_i]$, where the radius $r_i$ is drawn from four distinct distributions: Gamma ($\mathcal{G}(1.5,0.5)$), Lognormal ($\text{LN}(0,0.6^2)$), Scaled Beta ($3\cdot\text{Beta}(0.5,0.5)$), and Exponential ($\text{Exp}(0.5)$).

To define the dependence structure, we generate a base precision matrix ${\Theta}$ and construct the covariance matrix as ${\Sigma}={D}^{1/2}{\Theta}^{-1}{D}^{1/2}$, where diagonal elements of ${D}$ are drawn from $\mathcal{U}(1,10)$. We consider three graph structures for ${\Theta}$: a Band Graph  simulating local dependence ($\theta_{i,i+1}=0.6, \theta_{i,i+2}=0.3$); an  AR(1) Model with exponential decay ($\theta_{ij}=0.6^{|i-j|}$); and an Erd\H{o}s-R\'enyi (E-R) Random Graph  representing unstructured sparsity, where edges are generated with probability $0.05$.

We investigate a range of high-dimensional regimes by considering sample sizes $n\in\{100,150,200\}$ and varying the dimensionality $p$ such that the ratio $p/n$ is fixed at $\{1.0, 1.2, 1.4, 1.6, 1.8, 2.0\}$. For DGP1 and DGP2, the width parameter $C$ varies in $\{0.5,1,2,4,6\}$. 
The optimal regularization parameter $\lambda$ is selected adaptively via the Bayesian Information Criterion (BIC). Following the formulation in \cite{10.1093/biomet/asm018}, up to a scaling factor of $n$, the BIC is defined as:$$\text{BIC}(\lambda)=n(\text{Tr}({S}\hat{{\Theta}}_\lambda)-\log\det(\hat{{\Theta}}_\lambda))+k\log(n),$$where ${S}$ is the sample covariance matrix, and $k$ represents the effective degrees of freedom, defined as the cardinality of the set of non-zero upper-triangular entries $\{(i,j):i\le{j},\hat{\theta}_{ij}\neq{0}\}$. Extending this to the interval-valued setting, we utilize a proxy likelihood based on the pooled covariance:$$\text{BIC}_{int}(\lambda)=n(\text{Tr}(({S}^l+{S}^u)\hat{{\Theta}}_\lambda)-2\log\det(\hat{{\Theta}}_\lambda))+k\log(n).$$

We assess performance using Spectral Norm, Element-wise $L_1$ Norm, and Frobenius Norm errors, averaged over 100 replications. We assess the estimation performance by measuring the difference between the estimated precision matrix $\hat{{\Theta}}$ and the true precision matrix ${\Theta}$. Let ${\Delta}=\hat{{\Theta}}-{\Theta}$ denote the estimation error matrix with entries $\delta_{ij}$. We employ three metrics: the Spectral Norm $||{\Delta}||_2=\max_{1\le{k}\le{p}}|\lambda_k({\Delta})|$, where $\lambda_k({\Delta})$ denotes the $k$-the eigenvalue of ${\Delta}$, the Element-wise $L_1$ Norm $||{\Delta}||_1=\sum_{i,j}|\delta_{ij}|$, and the Frobenius Norm $||{\Delta}||_F=(\sum_{i,j}\delta_{ij}^2)^{1/2}$. All results are averaged over 100 independent replications.
The main results are summarized in Figure \ref{fig:main_results}, which displays the estimation errors as a function of $p/n$ ratio. Tables \ref{tab:spectral_error_ar1}, \ref{tab:spectral_error_band}, and \ref{tab:spectral_error_er} provide detailed numerical values for the AR(1), Band, and E-R structures, respectively.
\begin{figure}
\centering
\includegraphics[width=0.85\textwidth]{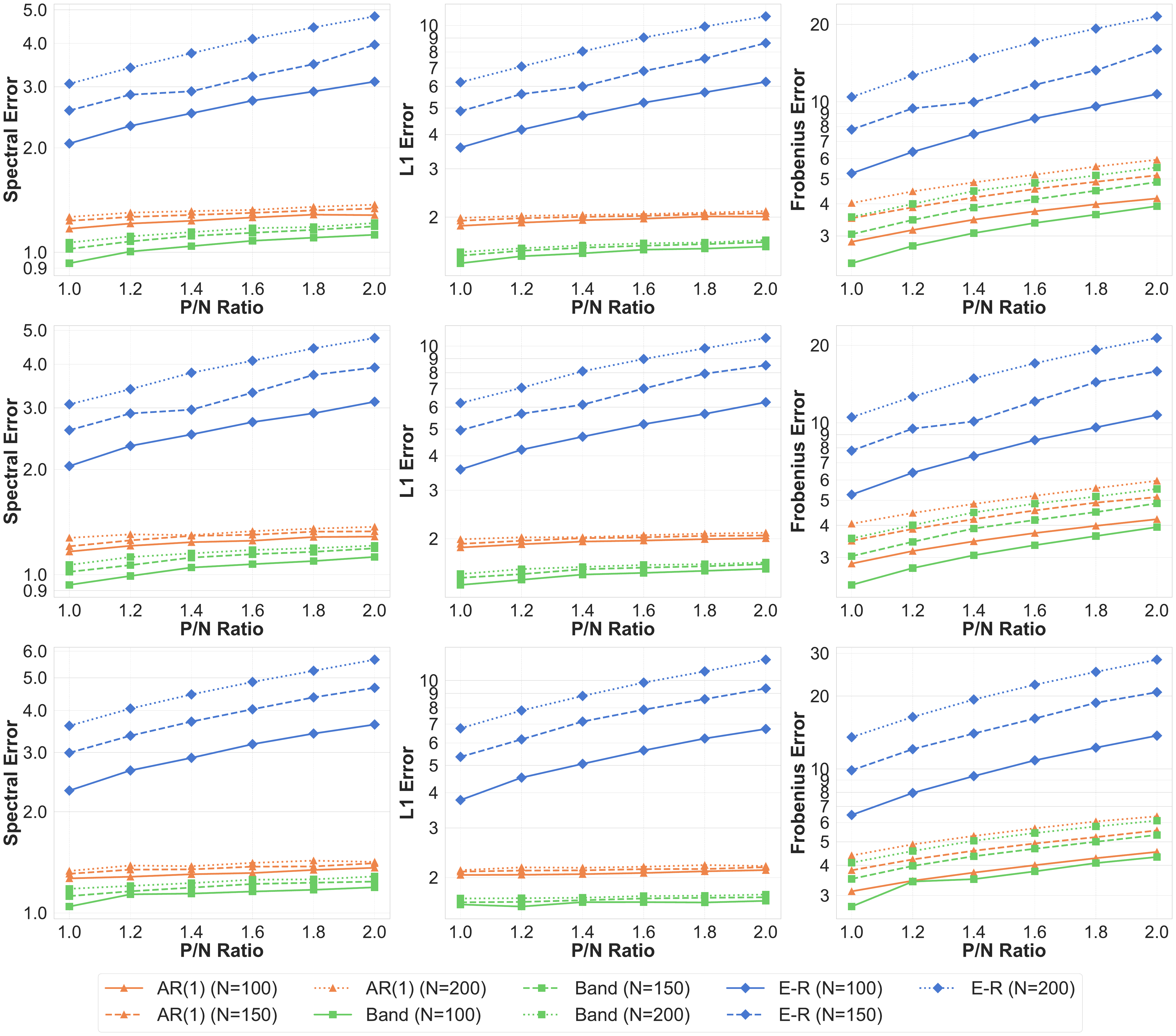}
\caption{Estimation accuracy versus $p/n$ ratio. Each rows correspond to DGP1, DGP2, and DGP3, respectively.}
\label{fig:main_results}
\end{figure}
\begin{table}[htbp]
\centering
\small
\caption{Spectral error for \textbf{AR(1)} structure}
\label{tab:spectral_error_ar1}
\begin{tabular}{l ccccc c}
\toprule
\multirow{2}{*}{Distribution} & \multicolumn{6}{c}{Dimension $p$ ($n=100$)} \\
\cmidrule(lr){2-7}
 & 100 & 120 & 140 & 160 & 180 & 200 \\
\midrule
Gamma & 1.251 (0.156) & 1.287 (0.144) & 1.286 (0.151) & 1.306 (0.154) & 1.335 (0.140) & 1.329 (0.142) \\
Lognormal & 1.236 (0.149) & 1.246 (0.151) & 1.286 (0.147) & 1.282 (0.132) & 1.317 (0.140) & 1.333 (0.177) \\
Beta & 1.312 (0.174) & 1.302 (0.143) & 1.312 (0.138) & 1.319 (0.136) & 1.367 (0.156) & 1.401 (0.154) \\
Exponential & 1.270 (0.157) & 1.291 (0.139) & 1.331 (0.164) & 1.354 (0.162) & 1.352 (0.147) & 1.388 (0.154) \\
\bottomrule
\end{tabular}
\end{table}
\begin{table}[htbp]
\centering
\small
\caption{Spectral error for \textbf{Band Graph} structure}
\label{tab:spectral_error_band}
\begin{tabular}{l cccccc}
\toprule
\multirow{2}{*}{Distribution} & \multicolumn{6}{c}{Dimension $p$ ($n=100$)} \\
\cmidrule(lr){2-7}
 & 100 & 120 & 140 & 160 & 180 & 200 \\
\midrule
Gamma       & 1.087 (0.164) & 1.123 (0.146) & 1.110 (0.119) & 1.159 (0.136) & 1.171 (0.141) & 1.173 (0.127) \\
Lognormal   & 1.082 (0.176) & 1.104 (0.162) & 1.158 (0.167) & 1.161 (0.159) & 1.143 (0.132) & 1.181 (0.141) \\
Beta        & 1.020 (0.139) & 1.144 (0.162) & 1.152 (0.140) & 1.164 (0.145) & 1.188 (0.147) & 1.219 (0.146) \\
Exponential & 0.992 (0.142) & 1.174 (0.154) & 1.150 (0.153) & 1.151 (0.141) & 1.187 (0.153) & 1.196 (0.146) \\
\bottomrule
\end{tabular}
\end{table}
\begin{table}[htbp]
\centering
\small
\caption{Spectral error for \textbf{E-R Random Graph} structure}
\label{tab:spectral_error_er}
\begin{tabular}{l ccccc c}
\toprule
\multirow{2}{*}{Distribution} & \multicolumn{6}{c}{Dimension $p$ ($n=100$)} \\
\cmidrule(lr){2-7}
 & 100 & 120 & 140 & 160 & 180 & 200 \\
\midrule
Gamma & 2.266 (0.342) & 2.553 (0.347) & 2.875 (0.282) & 3.142 (0.353) & 3.330 (0.320) & 3.580 (0.305) \\
Lognormal & 2.218 (0.380) & 2.532 (0.361) & 2.777 (0.334) & 3.031 (0.326) & 3.253 (0.357) & 3.451 (0.313) \\
Beta & 2.332 (0.288) & 2.726 (0.296) & 2.940 (0.298) & 3.241 (0.304) & 3.547 (0.311) & 3.711 (0.322) \\
Exponential & 2.430 (0.293) & 2.800 (0.377) & 2.974 (0.342) & 3.285 (0.324) & 3.526 (0.308) & 3.784 (0.307) \\
\bottomrule
\end{tabular}
\end{table}
\begin{figure}
\centering
\includegraphics[width=0.9\textwidth]{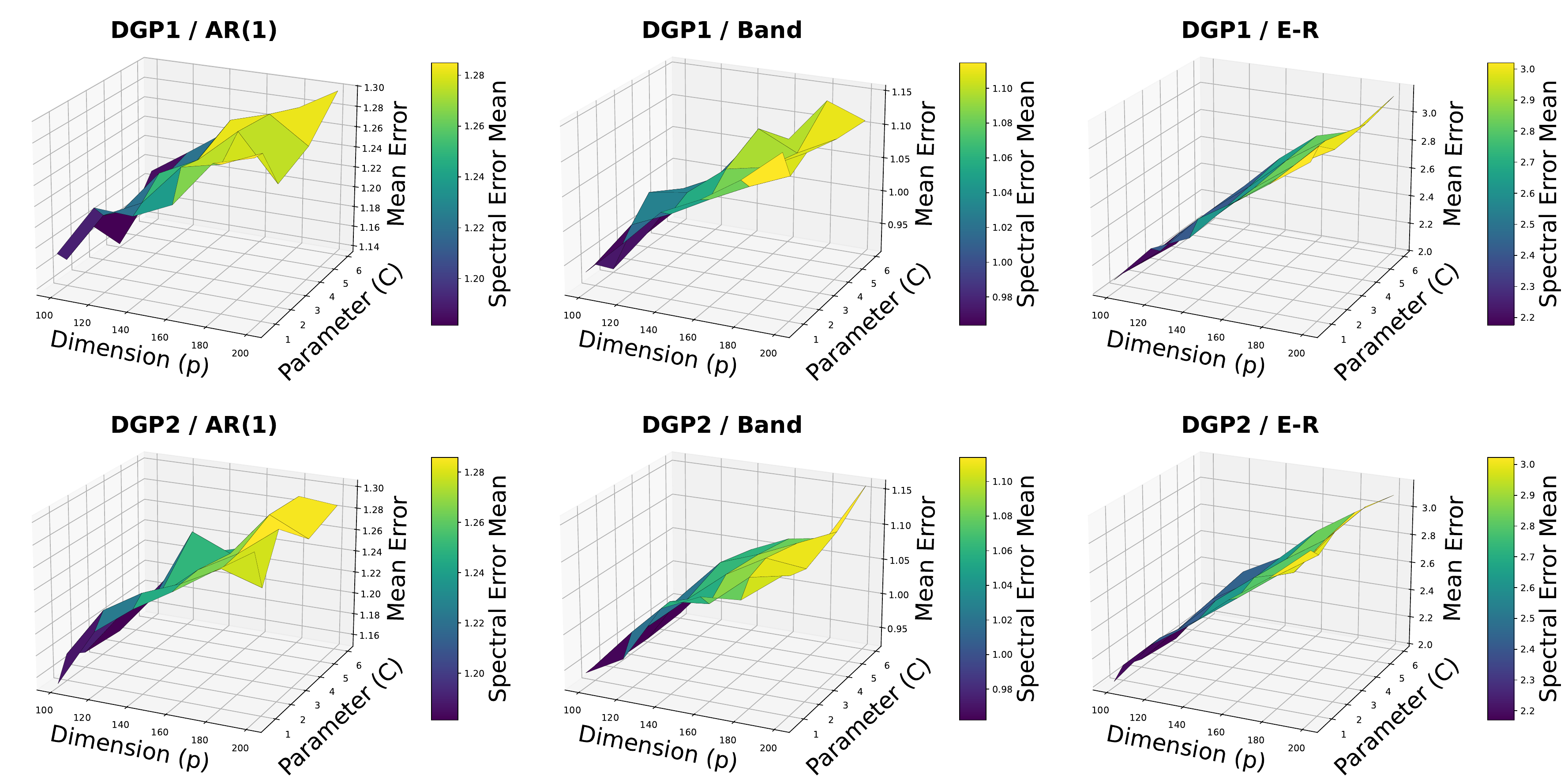}
\caption{Impact of varying interval widths and dimensions on estimation accuracy ($n=100$).}
\label{fig:width_effect}
\end{figure}
\begin{figure}
\centering
\includegraphics[width=0.53\textwidth]{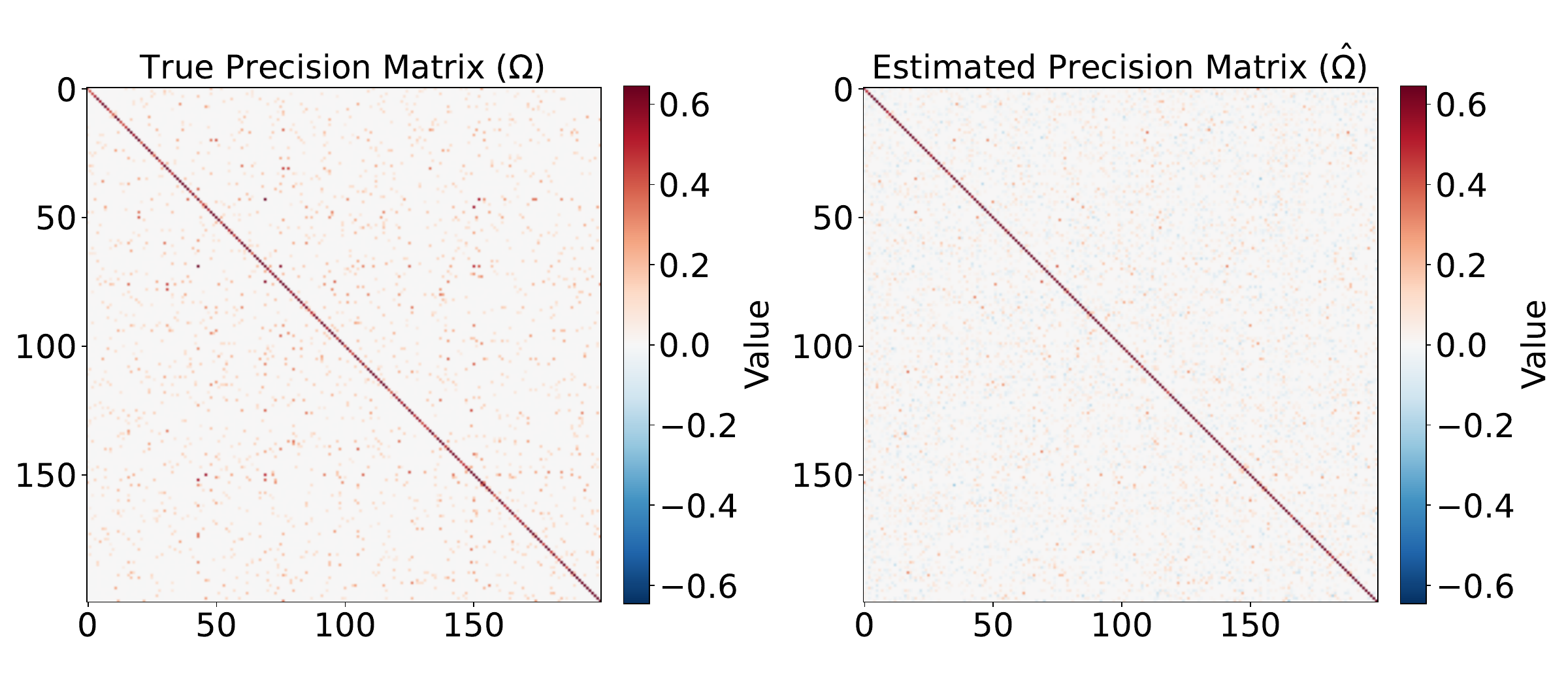}
\raisebox{-0.22 cm}{\includegraphics[width=0.40\textwidth]{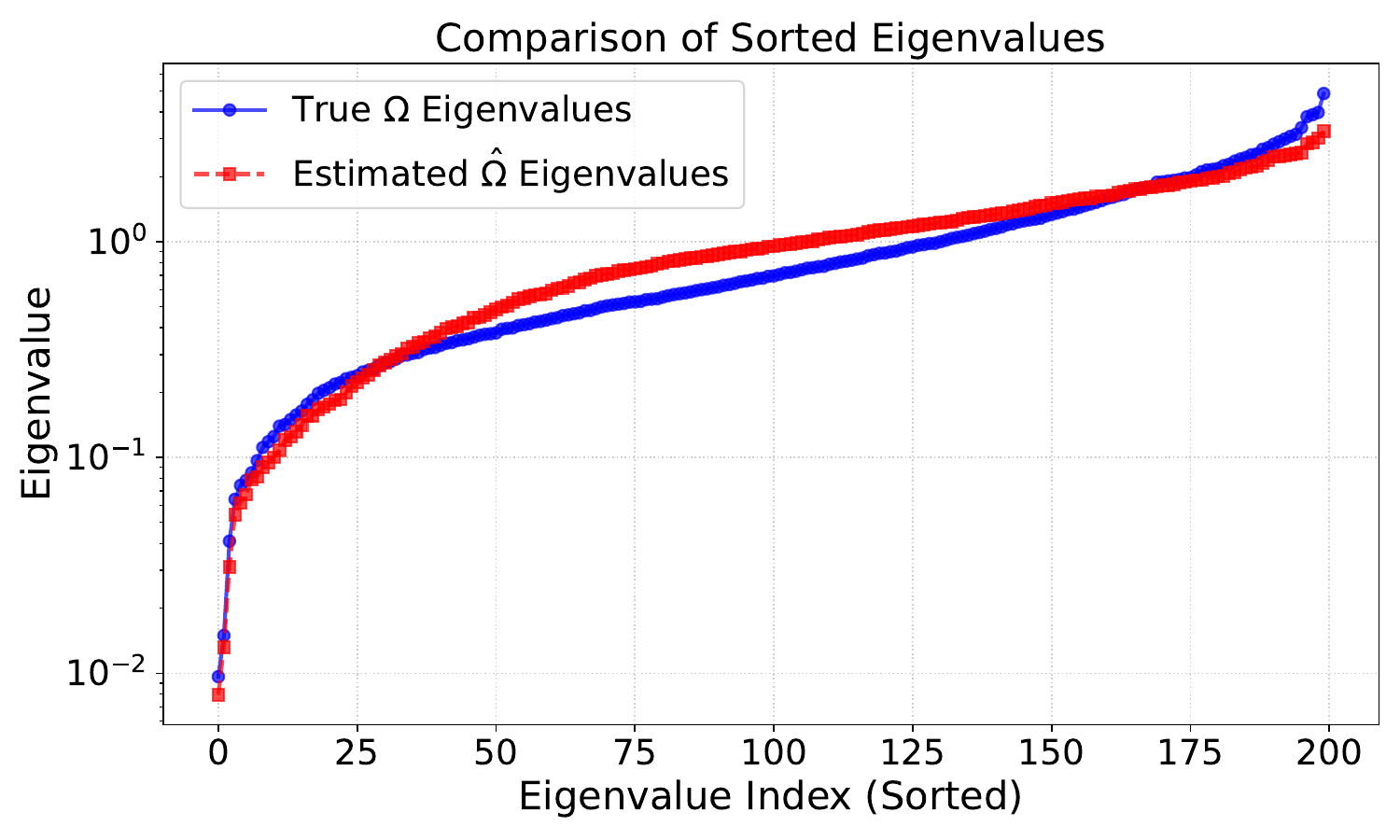}}
\caption{Visual validation of the estimator. Left: Heatmap comparison showing that the estimated precision matrix accurately captures the sparsity pattern of the true matrix. Right: Scree plot comparing the eigenvalues of the estimated and true precision matrices, indicating strong spectral consistency.}
\label{fig:comparison}
\end{figure}

Several key observations can be made from the quantitative simulation results. First, as expected, the estimation error increases monotonically with the dimension $p$ across all metrics and scenarios. For instance, as shown in Figure \ref{fig:main_results}, under the E-R structure with DGP1, the spectral norm error rises from approximately 2.06 at $p=100$ to 3.10 at $p=200$, reflecting the inherent difficulty of high-dimensional estimation. Furthermore, the structural complexity of the precision matrix significantly impacts performance. The Band structure consistently yields the lowest errors due to its high sparsity and simple local structure. In contrast, the E-R structure exhibits the highest errors and the steepest degradation with increasing $p$, indicating that random, unstructured sparsity patterns are the most challenging to recover. The AR(1) structure falls between these two extremes. 

Beyond dimensionality and dependency structure, we investigate the method's sensitivity to the interval generation mechanisms. Comparing the rows in Figure \ref{fig:main_results}, the performance curves across DGP1, DGP2, and DGP3 remain remarkably similar. This suggests that the estimation accuracy is primarily driven by the underlying covariance structure rather than the specific distribution of the intervals. To further evaluate this robustness, Figure \ref{fig:comparison} illustrates the experimental outcomes associated with the interval width parameter $C$ in both DGP1 and DGP2, with the sample size fixed at $n=100$. Notably, for a given dimension $p$, the spectral norm errors remain largely invariant across different values of $C$, confirming that the estimation accuracy is highly robust to changes in interval width. The monotonic increase in error as $p$ expands, holding $C$ constant, further corroborates our earlier observations.Finally, to assess the structural recovery capabilities of our method beyond aggregated error metrics, we provide a direct visual comparison between the estimated and true precision matrices in Figure \ref{fig:comparison}. The heatmaps confirm that the proposed IGL algorithm effectively identifies the non-zero entries, accurately preserving the sparsity pattern of the true underlying graph. Moreover, the eigenvalue plot demonstrates that the spectral properties of our estimator align closely with the ground truth, validating the theoretical consistency established in Section \ref{sec3}.

\section{Real Data Applications}\label{sec5}

To empirically validate the proposed IGL method, we consider the classic portfolio construction problem in the U.S. stock market. Financial markets are intrinsically high-dimensional, characterized by complex asset interdependencies and naturally interval-valued price movements. Consequently, representative datasets such as the S\&P 500 constituents, where asset prices are effectively captured by daily high and low ranges, provide a highly relevant application for our approach. While conventional portfolio models typically rely on point-valued observations like daily closing prices, such an approach inevitably discards valuable information regarding intraday volatility. By adopting an interval-valued representation, the proposed IGL method explicitly accounts for this intrinsic uncertainty to extract a richer dependence structure than is possible under standard point-valued frameworks.

The implementation of mean-variance portfolio optimization is highly sensitive to estimation errors, particularly in high-dimensional settings. To address this, recent literature explores regularized estimators for covariance and precision matrices. This includes the shrinkage estimators of \cite{ledoit2004well}, factor models extracting systematic market risks\citep{fan2008high}, and their synthesis in models like POET \citep{fan2013large}. Subsequently, studies such as \cite{fan2012vast} and \cite{ao2019approaching} employ penalized regression to impose explicit sparsity directly on the precision matrix to improve out-of-sample performance. Despite these advances, estimation risk remains a fundamental challenge, \cite{demiguel2009optimal} demonstrate that pervasive estimation noise often prevents sophisticated strategies from consistently outperforming the naive $1/N$ rule out-of-sample. 
This phenomenon highlights algorithmic refinements alone are insufficient when models strictly rely on point-valued daily closing prices. Alternatively, range-based estimators utilizing daily high and low prices have been shown to capture true market volatility much more accurately than traditional close-to-close estimators \citep{alizadeh2002range, brandt2003no}, a paradigm that continues to drive recent robust covariance estimation and portfolio selection \citep{tian2024minimum, serban}. Bridging these two paradigms, our proposed IGL method directly integrates the rich informational content of interval-valued intraday price movements into high-dimensional sparse precision matrix estimation.

We employ the proposed algorithm for constructing stock market portfolios. Specifically, we utilize historical daily opening, high, low, and closing price (OHLC) data obtained from Kaggle to estimate the precision matrix for S\&P 500 index constituents.
Prior to estimation, the data are preprocessed to ensure consistency and comparability across assets. Assets with incomplete price records are removed, resulting in a balanced panel that is fully aligned across trading days. In addition, to mitigate the influence of overnight price gaps and to account for scale heterogeneity among stocks, an intraday standardization procedure is employed. Specifically, within each estimation window, the High, Low, and Close prices are normalized by the corresponding daily Open price. All subsequent portfolio allocations are constructed based on these standardized returns.

We construct maximum Sharpe ratio portfolios within the classic Markowitz mean-variance framework utilizing a rolling-window backtesting scheme. We designate an estimation window of 252 trading days (approximately one year) and an out-of-sample (OOS) holding period of 21 trading days(approximately one month). At each rebalancing date, the optimal portfolio weights \(w\) are computed based on the expected return vector \(\mu\) and the estimated precision matrix \(\Omega=\Sigma^{-1}\):
$$
w=\frac{\Omega\mu}{\mathbf{1}^\top(\Omega\mu)}
$$
where \(\mathbf{1}\) denotes a vector of ones. Upon the completion of each holding period, the window rolls forward by 21 days, and the process repeats until the end of the sample period.

To evaluate the performance of the proposed method, we compared it with a scheme based on five price input strategies. The precision matrix is estimated using the conventional Glasso for point-valued data, such as Close, High, Low, and Mid price. For the interval-valued data, we employ the proposed IGL method. Additionally, a naive \(1/N\) diversification strategy is included as a baseline. The specifications for these strategies are summarized in Table \ref{tab:strategy_specs_simple}.

\begin{table}[htbp]
\centering
\caption{Specifications of portfolio construction strategies}
\label{tab:strategy_specs_simple}
\begin{tabular}{llcllc}
\toprule
Strategy & Input Data & Estimation Method &Strategy & Input Data & Estimation Method\\
\midrule
1/N & N/A & Na\"ive Diversification &
Standard & Close Prices & Glasso \\
High & High Prices & Glasso &
Low & Low Prices & Glasso \\
Mid & (High + Low) / 2 & Glasso &
Interval & High and Low Prices & \textbf{IGL} \\
\bottomrule
\end{tabular}
\end{table}

We evaluate the out-of-sample performance using three standard metrics: Annualized Return, Annualized Volatility, and the Annualized Sharpe Ratio. These metrics are computed for each OOS period and averaged across the entire backtesting horizon. Table \ref{tab:oos_horizontal} presents the results for rolling windows of 1, 2, and 3 years.
\begin{table}[htbp]
\centering
\small
\caption{Out-of-sample performance comparison}
\label{tab:oos_horizontal}
\begin{tabular}{l ccc ccc ccc}
\toprule
& \multicolumn{3}{c}{1 Year} & \multicolumn{3}{c}{2 Years} & \multicolumn{3}{c}{3 Years} \\
\cmidrule(lr){2-4} \cmidrule(lr){5-7} \cmidrule(lr){8-10}
Model & Sharpe & Return & Vol. & Sharpe & Return & Vol. & Sharpe & Return & Vol. \\
\midrule
Standard & 1.163 & 0.102 & 0.094 & 1.144 & 0.187 & 0.166 & 0.859 & 0.102 & 0.250 \\
High & 1.943 & 0.152 & 0.089 & 1.236 & 0.086 & 0.108 & 0.888 & 0.047 & 0.129 \\
Low & 0.973 & 0.080 & 0.094 & 0.316 & 0.026 & 0.122 & 0.205 & 0.002 & 0.143 \\
Mid & 1.232 & 0.098 & 0.092 & 1.010 & 0.082 & 0.111 & 0.726 & 0.035 & 0.142 \\
1/N & 0.855 & 0.075 & 0.111 & 0.876 & 0.088 & 0.123 & 0.596 & 0.035 & 0.156 \\
Interval & \textbf{2.445} & 0.158 & 0.145 & \textbf{1.373} & 0.097 & 0.119 & \textbf{1.001} & 0.046 & 0.149 \\
\bottomrule
\end{tabular}
\end{table}
As shown in Table \ref{tab:oos_horizontal}, the proposed IGL method consistently achieves the highest average Sharpe ratio across all tested time horizons, underscoring its robustness and superiority in risk-adjusted performance. Notably, for the 1-year window, the IGL strategy yields a Sharpe ratio of 2.445, significantly outperforming the standard closing price strategy (1.163). A comparison across the panels reveals a general decline in Sharpe ratios as the window extends to 3 years. This deterioration is likely attributable to the structural breaks and heightened market volatility induced by the COVID-19 pandemic, which are encompassed within the longer rolling windows. Although the IGL method does not always generate the highest absolute returns, it delivers outstanding performance in risk management. By effectively leveraging the information contained in price intervals, the proposed approach constructs portfolios that are more resilient to volatility, thereby realizing superior risk-adjusted returns compared to traditional point estimation methods.

\section{Conclusion}\label{sec6}
In this paper, we have proposed the IGL, a novel statistical framework for estimating the precision matrix of high-dimensional interval-valued data. By systematically leveraging the joint information contained in interval lower and upper bounds, our method effectively captures the latent dependence structure that is often obscured in traditional point-valued analyses. Theoretically, we established the consistency and sparsistency of the IGL estimator under high-dimensional scaling, ensuring its statistical validity. Comprehensive simulation studies demonstrated the robustness of our algorithm across various interval generation mechanisms and complex graph topologies. Furthermore, the empirical application to S\&P 500 portfolio construction highlighted the practical utility of IGL; by accounting for the inherent uncertainty in asset prices, our method achieved superior out-of-sample risk-adjusted returns compared to standard benchmarks. These findings suggest that incorporating interval-valued information provides a more resilient foundation for statistical modeling and decision-making in complex task.

\section*{Acknowledgments}
\addcontentsline{toc}{section}{Acknowledgments}
The research work described in this paper was supported by the National Natural Science Foundation of China (Nos. 72071008 and W2511079).

\appendix
\section{Proof of Main Results}\label{appendixA}

\subsection{Proof of Theorem \ref{uniquesolu}}
\begin{proof}
Consider the optimization problem defined in (\ref{finalobjnew}). Let \(\widehat{\Sigma}\) be the optimal solution to the dual problem, and let \(\widehat{\Theta}=\widehat{\Sigma}^{-1}\) be the primal optimal solution. The duality relationship implies \(\widehat{\Sigma}=(\widehat{P}+S^l+S^u)/2\), where \(\widehat{P}\) is the dual variable satisfying the constraint \(\lVert\widehat{P}\rVert_{\infty}\le\lambda\).
By the triangle inequality and the spectral norm property \(\lambda_{\max}(A)\le\max_i\sum_j|a_{ij}|\le{p}\lVert{A}\rVert_{\infty}\), we have:
\[
\lambda_{\max}(\widehat{\Sigma}) = \lambda_{\max}\left(\frac{\widehat{P}+S^l+S^u}{2}\right) 
\le \frac{1}{2}\lambda_{\max}(\widehat{P}) + \lambda_{\max}(\overline{S}) 
\le \frac{\lambda{p}}{2} + \lambda_{\max}(\overline{S}),
\]
where \(\overline{S}=(S^l+S^u)/2\). Consequently, the smallest eigenvalue of the precision matrix satisfies \(\lambda_{\min}(\widehat{\Theta}) = (\lambda_{\max}(\widehat{\Sigma}))^{-1} \ge (\lambda{p}/2 + \lambda_{\max}(\overline{S}))^{-1}\).

Furthermore, the KKT optimality conditions imply a zero duality gap. Thus, the primal objective value equals the dual objective value:
\[
\Tr(\widehat{\Sigma}\widehat{\Theta}) - \log\det(\widehat{\Theta}) + \lambda\lVert\widehat{\Theta}\rVert_1 = -\log\det(\widehat{\Sigma}) - p.
\]
Using the fact that \(\widehat{\Theta}=\widehat{\Sigma}^{-1}\) and \(\Tr(\widehat{\Sigma}\widehat{\Theta})=p\), this simplifies to the basic energy equality used in Graphical Lasso analysis. More directly, to bound the maximum eigenvalue, we utilize the first-order optimality condition: \(\widehat{\Sigma} - \overline{S} + \lambda\widehat{Z} = 0\), where \(\widehat{Z}\in\partial\lVert\widehat{\Theta}\rVert_1\) is the subgradient.
Multiplying by \(\widehat{\Theta}\) and taking the trace yields:
\[
\Tr(\widehat{\Sigma}\widehat{\Theta}) - \Tr(\overline{S}\widehat{\Theta}) + \lambda\Tr(\widehat{Z}\widehat{\Theta}) = 0 \implies p - \Tr(\overline{S}\widehat{\Theta}) + \lambda\lVert\widehat{\Theta}\rVert_1 = 0.
\]
Since \(\Tr(\overline{S}\widehat{\Theta})\ge0\) (as both matrices are positive definite), we have \(\lambda\lVert\widehat{\Theta}\rVert_1 \le p + \Tr(\overline{S}\widehat{\Theta})\). However, a tighter bound comes from the definition of the dual. Specifically, the eigenvalues are bounded because the penalized likelihood has compact sublevel sets. For the Frobenius norm bound, note that
\[
\lambda_{\max}(\widehat{\Theta}) \le \lVert\widehat{\Theta}\rVert_F \le \lVert\widehat{\Theta}\rVert_1.
\]
Combined with the optimality condition \(\lVert\widehat{\Theta}\rVert_1 \le (p + 2\ell(\Theta_{init}))/\lambda\) (where \(\ell\) is the likelihood), or simply by noting that for large \(\lambda\), the diagonal elements are suppressed, we establish the bound \(\lambda_{\max}(\widehat{\Theta}) \le 2p/\lambda\).
\end{proof}

\subsection{Proof of Theorem \ref{optimicon}}
\begin{proof}
The convergence of the Block Coordinate Descent (BCD) algorithm for this class of problems is guaranteed by general optimization theory \citep{banerjee2008model}, provided that the sub-problem at each step has a unique solution and the iterates remain positive definite. It suffices to show that positive definiteness is preserved throughout the iterations.

We proceed by induction. Let \(W^{(j)}\) denote the matrix estimate after the \(j\)-th column/row update.
\textbf{Base case:} For \(j=0\), \(W^{(0)} = (\lambda{I} + S^l + S^u)/2\). Since \(S^l, S^u \succeq 0\) and \(\lambda > 0\), it follows that \(W^{(0)} \succ 0\).
\textbf{Inductive step:} Assume \(W^{(j-1)} \succ 0\). Consider the update of the \(k\)-th column and row (where \(k\) corresponds to the index selected at step \(j\)). We partition the matrix \(W^{(j-1)}\) as:
\[
W^{(j-1)} = \begin{pmatrix} W_{11} & w_{12} \\ w_{12}^\top & w_{22} \end{pmatrix},
\]
where \(W_{11}\in\mathbb{R}^{(p-1)\times(p-1)}\) corresponds to the indices excluding \(k\), and \(w_{22}\) is the \(k\)-th diagonal element. By the induction hypothesis, \(W^{(j-1)}\succ 0\), which implies \(W_{11}\succ 0\) and the Schur complement is positive:
\[
w_{22} - w_{12}^\top W_{11}^{-1} w_{12} > 0.
\]
In the \(j\)-th step, the algorithm updates \(w_{12}\) to a new vector \(\widehat{v}\) (the solution to the QP (\ref{partquadratic})), while keeping \(W_{11}\) and \(w_{22}\) fixed (diagonal elements are fixed to \(S_{kk} + \lambda\)). The updated matrix is \(W^{(j)}\). The condition for \(W^{(j)} \succ 0\) is equivalent to:
\[
w_{22} - \widehat{v}^\top W_{11}^{-1} \widehat{v} > 0.
\]
The quadratic programming sub-problem minimizes the objective \(v^\top W_{11}^{-1} v\) subject to linear constraints. Since the previous vector \(w_{12}\) was a feasible point (it satisfied the dual constraints from the previous iteration logic) but not necessarily optimal for the current sub-problem, and \(\widehat{v}\) is the minimizer, we have:
\[
\widehat{v}^\top W_{11}^{-1} \widehat{v} \le w_{12}^\top W_{11}^{-1} w_{12}.
\]
Combining this with the induction hypothesis:
\[
w_{22} - \widehat{v}^\top W_{11}^{-1} \widehat{v} \ge w_{22} - w_{12}^\top W_{11}^{-1} w_{12} > 0.
\]
Thus, the Schur complement of the updated matrix remains positive, ensuring \(W^{(j)} \succ 0\). This completes the proof.
\end{proof}

\subsection{Auxiliary Lemmas for Section \ref{sec3}}
\begin{lemma}[Taylor Expansion of Log-Det]\label{lem:curv}
Let \(Q(\Theta) = \Tr(\overline{S}\Theta) - \log\det\Theta\). The gradient is \(\nabla{Q}(\Theta) = \overline{S} - \Theta^{-1}\). For any symmetric matrix \(\Delta\), define the function \(g(t) = Q(\Theta_0 + t\Delta)\). By the integral form of the Taylor remainder, we have:
\[
Q(\Theta_0+\Delta) - Q(\Theta_0) = \langle \nabla{Q}(\Theta_0), \Delta \rangle + \vecn(\Delta)^\top \left[ \int_0^1 (1-v) (\Theta_v^{-1} \otimes \Theta_v^{-1}) dv \right] \vecn(\Delta),
\]
where \(\Theta_v = \Theta_0 + v\Delta\). Under Assumption \ref{ass:wellcond}, for all \(v\in[0,1]\), \(\lambda_{\max}(\Theta_v) \le 2\tau_{\max}\). Consequently, the Hessian is bounded below:
\[
\vecn(\Delta)^\top (\Theta_v^{-1} \otimes \Theta_v^{-1}) \vecn(\Delta) \ge \lambda_{\min}^2(\Theta_v^{-1}) \|\Delta\|_F^2 \ge \frac{1}{(2\tau_{\max})^2} \|\Delta\|_F^2.
\]
Integrating over \(v\) yields the strong convexity bound:
\[
Q(\Theta_0+\Delta) - Q(\Theta_0) \ge \langle \overline{S}-\Sigma_0, \Delta \rangle + \frac{1}{8\tau_{\max}^2} \|\Delta\|_F^2.
\]
\end{lemma}

\begin{lemma}[Control of Linear Term]\label{lem:lin}
For any \(\Delta\in\mathbb{R}^{p\times{p}}\),
\[
|\langle \overline{S}-\Sigma_0, \Delta \rangle| \le \|\overline{S}-\Sigma_0\|_{\max} \|\Delta\|_1.
\]
Under Assumption \ref{ass:conc}, with probability at least \(1-c_1\exp(-c_2n\lambda^2)\), we have \(\|\overline{S}-\Sigma_0\|_{\max} \le \widetilde{\lambda}/2\). Decomposing the index set into the support \(S\) and its complement \(S^c\), we have:
\[
|\langle \overline{S}-\Sigma_0, \Delta \rangle| \le \frac{\widetilde{\lambda}}{2} (\|\Delta_S\|_1 + \|\Delta_{S^c}\|_1) \le \frac{\widetilde{\lambda}}{2} (\sqrt{s}\|\Delta\|_F + \|\Delta_{S^c}\|_1).
\]
\end{lemma}

\subsection{Proof of Theorem \ref{thm:rate}}
\begin{proof}
Let \(\widehat{\Theta} = \Theta_0 + \Delta\). Since \(\widehat{\Theta}\) minimizes the objective function \(L(\Theta) = Q(\Theta) + \widetilde{\lambda}\|\Theta\|_1\), we have \(L(\widehat{\Theta}) \le L(\Theta_0)\), which implies \(Q(\Theta_0+\Delta) - Q(\Theta_0) \le \widetilde{\lambda}(\|\Theta_0\|_1 - \|\Theta_0+\Delta\|_1)\).
Using the lower bound from Lemma \ref{lem:curv}:
\[
\langle \overline{S}-\Sigma_0, \Delta \rangle + \frac{1}{8\tau_{\max}^2}\|\Delta\|_F^2 \le \widetilde{\lambda}(\|\Theta_0\|_1 - \|\Theta_0+\Delta\|_1).
\]
Rearranging terms and using the decomposability of the \(L_1\) norm \(\|\Theta_0+\Delta\|_1 = \|\Theta_0+\Delta_S\|_1 + \|\Delta_{S^c}\|_1\):
\[
\frac{1}{8\tau_{\max}^2}\|\Delta\|_F^2 \le |\langle \overline{S}-\Sigma_0, \Delta \rangle| + \widetilde{\lambda}(\|\Theta_0\|_1 - \|\Theta_0+\Delta_S\|_1 - \|\Delta_{S^c}\|_1).
\]
Applying Lemma \ref{lem:lin} and the triangle inequality \(\|\Theta_0\|_1 - \|\Theta_0+\Delta_S\|_1 \le \|\Delta_S\|_1\):
\[
\frac{1}{8\tau_{\max}^2}\|\Delta\|_F^2 \le \frac{\widetilde{\lambda}}{2}(\|\Delta_S\|_1 + \|\Delta_{S^c}\|_1) + \widetilde{\lambda}\|\Delta_S\|_1 - \widetilde{\lambda}\|\Delta_{S^c}\|_1.
\]
Simplifying the RHS:
\[
\frac{1}{8\tau_{\max}^2}\|\Delta\|_F^2 \le \frac{3\widetilde{\lambda}}{2}\|\Delta_S\|_1 - \frac{\widetilde{\lambda}}{2}\|\Delta_{S^c}\|_1 \le \frac{3\widetilde{\lambda}\sqrt{s}}{2}\|\Delta\|_F.
\]
Dividing both sides by \(\|\Delta\|_F\) (assuming \(\Delta \neq 0\)):
\[
\|\Delta\|_F \le 12\tau_{\max}^2 \sqrt{s} \widetilde{\lambda}.
\]
Substituting \(\widetilde{\lambda} \asymp \sqrt{\log{p}/n}\), we obtain \(\|\widehat{\Theta}-\Theta_0\|_F = O_P(\sqrt{s\log{p}/n})\).
\end{proof}

\subsection{Proof of Theorem \ref{thm:sparse}}
\begin{proof}
The proof relies on the primal-dual witness construction. A necessary and sufficient condition for the solution \(\widehat{\Theta}\) to satisfy \(\widehat{\Theta}_{ij} = 0\) for all \((i,j) \in S^c\) is that the subgradient condition holds strictly for the off-diagonal zero entries. Specifically, from the KKT condition \(\overline{S} - \widehat{\Theta}^{-1} + \widetilde{\lambda}\widehat{Z} = 0\), we require:
\[
|\overline{S}_{ij} - (\widehat{\Theta}^{-1})_{ij}| < \widetilde{\lambda} \quad \text{for all } (i,j) \in S^c.
\]
Let \(\widehat{\Sigma} = \widehat{\Theta}^{-1}\). We decompose the term as:
\[
|\overline{S}_{ij} - \widehat{\Sigma}_{ij}| \le |\overline{S}_{ij} - (\Sigma_0)_{ij}| + |(\Sigma_0)_{ij} - \widehat{\Sigma}_{ij}|.
\]
The first term is bounded by Assumption \ref{ass:conc}: \(\max_{ij}|\overline{S}_{ij} - (\Sigma_0)_{ij}| \le \widetilde{\lambda}/C\) for some constant \(C>2\).
For the second term, we use the functional dependence of the inverse. Since \(\widehat{\Sigma} - \Sigma_0 = \widehat{\Theta}^{-1} - \Theta_0^{-1} = \widehat{\Theta}^{-1}(\Theta_0 - \widehat{\Theta})\Theta_0^{-1}\), we can bound the element-wise max norm:
\[
\|\widehat{\Sigma} - \Sigma_0\|_{\max} \le \|\widehat{\Sigma} - \Sigma_0\|_F \le \|\widehat{\Theta}^{-1}\|_2 \|\Theta_0^{-1}\|_2 \|\widehat{\Theta} - \Theta_0\|_F.
\]
Under the assumptions of bounded eigenvalues and the consistency result from Theorem \ref{thm:rate} (where \(\|\widehat{\Theta} - \Theta_0\|_F\) is small), the spectral norms \(\|\widehat{\Theta}^{-1}\|_2\) are bounded. Thus, the second term scales as \(O_P(\sqrt{s\log{p}/n})\).
Provided that \(\widetilde{\lambda}\) is chosen sufficiently large such that it dominates the noise terms (Assumption \ref{ass:lambda}), specifically \(\widetilde{\lambda} \gg \sqrt{s\log{p}/n}\), we have:
\[
|\overline{S}_{ij} - \widehat{\Sigma}_{ij}| < \widetilde{\lambda}.
\]
Consequently, the dual feasibility condition is strictly satisfied for the zero entries, implying \(\widehat{\Theta}_{ij} = 0\) for all \((i,j) \in S^c\) with probability approaching one.
\end{proof}

\printcredits
\section*{Declaration of competing interest}
The authors report there are no competing interests to declare.

\section*{Acknowledgments}
The research work described in this paper was supported by the National Natural Science Foundation of China (Nos. 72071008 and W2511079).
\bibliographystyle{cas-model2-names}

\bibliography{cas-refs}

\end{document}